\documentclass{theoretics}

\title{Tight Lower Bound for Multicolor Discrepancy}

\ThCSauthor[google]{Pasin Manurangsi}{pasin@google.com}[0000-0002-1052-2801]
\ThCSaffil[google]{Google Research, Thailand}
\ThCSauthor[UCLA]{Raghu Meka}{raghum@cs.ucla.edu}[0009-0004-7676-2762]
\ThCSaffil[UCLA]{University of California, Los Angeles}
  
\ThCSthanks{Invited from SOSA 26~\cite{ManurangsiM26}.}
\ThCSshortnames{P.\ Manurangsi, R.\ Meka}  
\ThCSshorttitle{Tight Lower Bound for Multicolor Discrepancy}
\ThCSyear{2026}
\ThCSarticlenum{18}
\ThCSreceived{Mar 12, 2026}
\ThCSaccepted{Oct 18, 2026}
\ThCSpublished{Oct 30, 2026}
\ThCSdoicreatedtrue
\ThCSkeywords{Multicolor Discrepancy, Group Fair Division, Envy-Freeness, Proportionality, Consensus Division}

\usepackage{csquotes}

\newcommand{\N}{\mathbb{N}}
\newcommand{\bA}{\mathbf{A}}
\newcommand{\tbA}{\tilde{\bA}}

\newcommand{\bone}{\mathbf{1}}
\newcommand{\R}{\mathbb{R}}

\newcommand{\bx}{\mathbf{x}}

\newcommand{\bz}{\mathbf{z}}
\newcommand{\bH}{\mathbf{H}}
\newcommand{\bW}{\mathbf{W}}

\newcommand{\cprop}{c^{\mathrm{PROP}}}
\newcommand{\cef}{c^{\mathrm{EF}}}
\newcommand{\ccdiv}{c^{\mathrm{CD}}_k}

\DeclareMathOperator{\exdisc}{disc^{max}}
\DeclareMathOperator{\exodisc}{asymdisc^{max}}

\DeclareMathOperator{\disc}{disc}
\DeclareMathOperator{\odisc}{asymdisc}
\DeclareMathOperator{\wdisc}{wdisc}
\newcommand{\exwdisc}{\operatorname{wdisc}^{\operatorname{max}}}

\newcommand{\tchi}{\tilde{\chi}}
\newcommand{\ochi}{\overline{\chi}}
\newcommand{\Slarge}{S_{\mathrm{large}}}
\newcommand{\tSlarge}{\tilde{S}_{\mathrm{large}}}

\allowdisplaybreaks

\begin{document}

\maketitle

\begin{abstract}
We prove the following asymptotically tight lower bound for $k$-color discrepancy: For any $k \geq 2$, there exists a hypergraph with $n$ \emph{hyperedges} such that its $k$-color discrepancy is at least $\Omega(\sqrt{n})$. This improves on the previously known lower bound of $\Omega(\sqrt{n/\log k})$ due to Caragiannis et al.~\cite{caragiannis2025newlowerboundmulticolor}. As an application, we show that our result implies improved lower bounds for group fair division.
\end{abstract}

\section{Introduction}

In the classical combinatorial setting of discrepancy theory (see e.g.~\cite{Chazelle01}), we are given a hypergraph and the goal is to color each vertex using one of the \emph{two colors}, say red and blue. The discrepancy of a hyperedge is the difference between the number of its blue vertices (equivalently, red vertices) and half of the total number of vertices in the hyperedge\footnote{Sometimes the discrepancy is defined to be twice this amount, i.e.\ the difference between the numbers of its blue vertices and its red vertices}. The goal is to find a coloring that minimizes the corresponding discrepancy of the hypergraph--which is defined as the maximum discrepancy across all the hyperedges.

In this work, we focus on the \emph{$k$-color} generalization of discrepancy, pioneered by Doerr and Srivastav~\cite{DoerrSr03}. Here we color each vertex using one of $k$ colors, and the discrepancy of each hyperedge is the maximum difference between the number of its vertices of each color and $\frac{1}{k}$ of the total number of its vertices. For notational convenience, we formalize this using a matrix notion, where the matrix $\bA$ can be viewed as the incidence matrix of the hypergraph, as follows.

\begin{definition}[\cite{DoerrSr03}] \label{def:k-disc}
The \emph{$k$-color discrepancy} of a matrix\footnote{While this allows $\bA \in [0, 1]^{n \times m}$, our lower bound holds for $\bA \in \{0, 1\}^{n \times m}$, which corresponds to a hypergraph.} $\bA \in [0, 1]^{n \times m}$ is defined as\footnote{We use $\bone$ to denote the all-1 vector and $\bone(S)$ to denote the indicator vector of $S$.}
\begin{align*}
\disc(\bA, k) := \min_{\chi: [m] \to [k]} \max_{s \in [k]} \left\|\bA\left(\frac{1}{k} \cdot \bone - \bone(\chi^{-1}(s))\right)\right\|_{\infty}.
\end{align*}
\end{definition}
We are interested in a bound on the $k$-color discrepancy among all matrices $\bA$ with $n$ rows (corresponding to hypergraphs with $n$ \emph{hyperedges}\footnote{We use $n$ to denote the number of hyperedges to conform with the notations in fair division (e.g.~\cite{ManurangsiS22}).}), as formalized below.
\begin{definition}
For all $k, n \in \N$, let 
$\exdisc(n, k) := \sup_{m \in \N} \sup_{\bA \in [0, 1]^{n \times m}} \disc(\bA, k).$
\end{definition}

For the case of two colors ($k = 2$), the celebrated result of \cite{Spencer85}\footnote{While \cite{Spencer85} states the upper bound as $O(\sqrt{m \ln(2m/n)})$ when $n \leq m$, it is clear that this is at most $O(\sqrt{n})$. For an explicit statement, see \cite[Corollary 12.3.4]{AlonSp00}.} implies that $\exdisc(n, 2) \leq O(\sqrt{n})$.
Doerr and Srivastav~\cite{DoerrSr03} extended this bound\footnote{While \cite{DoerrSr03} focuses only on the case $m \geq n$, a \enquote{well known} technique \cite[Section 4]{Spencer85} (see also \cite{LovaszSpVe86}) can extend their bound to arbitrary $m$. See \cite[Lemma 2.7]{ManurangsiS22} and the explanation therein for more detail.} for all $k \in \N$. However, they only provide a lower bound of $\Omega(\sqrt{n/k})$, leaving a gap of $\Theta(\sqrt{k})$. In a recent work of Caragiannis et al.~\cite{caragiannis2025newlowerboundmulticolor}, an improved lower bound of $\Omega(\sqrt{n/\log k})$ is established, thus significantly reducing the gap to just $\Theta(\sqrt{\log k})$.

Caragiannis et al.\ also show that their probabilistic construction and argument can be extended to yield improved lower bounds for group fair division, which will be defined formally in \autoref{sec:group-fair}. Interestingly, unlike previous work of Manurangsi and Suksompong~\cite{ManurangsiS22} which obtains such lower bounds via direct reductions from discrepancy results, Caragiannis et al.\ use separate (slightly different) constructions and arguments for each of their results.

\paragraph*{Our Contributions.} First, we close the aforementioned gap on the bound for $\exdisc(n, k)$ by showing that it is $\Omega(\sqrt{n})$ regardless of the value of $k$ (\autoref{cor:main-multicolor}). Our proof is based on a very simple modification of the $k = 2$ case. In fact, our technique gives the same lower bound (\autoref{thm:main-weighted}) for a smaller quantity called \emph{$1/k$-weighted discrepancy} (defined in \autoref{sec:prelim}), which, as explained below, can be useful in subsequent applications.

Our second contribution is to give asymptotically improved bounds for group fair division. The improvement is at least a factor of $\Omega(\sqrt{\log k})$ in all regimes of parameters and, in some regimes of parameters, can be as large as $\Omega(k \sqrt{\log k})$. (See \autoref{sec:group-fair} for a more detailed discussion.) Unlike~\cite{caragiannis2025newlowerboundmulticolor}, we achieve these improvements via simple reductions from our lower bound on $1/k$-weighted discrepancy. These reductions can be viewed as strengthenings of those in \cite{ManurangsiS22}.

Finally, we give an improved \emph{upper bound} for an approximate fairness guarantee (\emph{proportionality}) for group fair division. Note that this result does \emph{not} use our main lower bound on discrepancy, but rather is based on a simple refinement of current upper bound proof techniques from~\cite{DoerrSr03,ManurangsiS22}. The full details are deferred to \autoref{sec:prop-ub}.

\section{Preliminaries}
\label{sec:prelim}

For $n \in \N$, let $[n]$ denote $\{1, \dots, n\}$.
For $x \in \R$, we write $[x]_+$ as a shorthand for $\max\{0, x\}$. We extend this notion naturally to vectors: for $\bx \in \R^d$, $[\bx]_+$ is a vector whose $i$-th entry is $[\bx_i]_+$.

We will also work with the notion of \emph{$p$-weighted discrepancy} for $p \in [0, 1]$. Roughly speaking, the goal here is to find a vector $\bx \in \{0, 1\}^m$ such that $\bA \cdot \bx$ is close to $p$ times $\bA \cdot \bone$, as formalized below. When $p = 1/2$, this notion coincides with 2-color discrepancy (by letting $\bx = \bone(\chi^{-1}(1))$). However, when $p = 1/k$ for $k > 2$, this is a weaker notion than $k$-color discrepancy, since $\bx$ only corresponds to a single color (e.g., $\bone(\chi^{-1}(1))$) and we impose no requirement on the remaining items (corresponding to $\bone - \bx$) in $p$-weighted discrepancy.

\begin{definition}[\cite{DoerrSr03}]
For $p \in [0, 1]$, the \emph{$p$-weighted discrepancy} of $\bA \in [0, 1]^{n \times m}$ is 
$$\wdisc_p(\bA) := \min_{\bx \in \{0, 1\}^m} \|\bA(p \cdot \bone - \bx)\|_{\infty}.$$

Furthermore, for $p \in [0,1]$ and $n \in \N$, let $\exwdisc_p(n) := \sup_{m \in \N} \sup_{\bA \in [0, 1]^{n \times m}} \wdisc_p(\bA)$.
\end{definition}

It is known that the $p$-weighted discrepancy is at most $O(\sqrt{n})$:
\begin{theorem}[\cite{LovaszSpVe86}] \label{thm:wdisc-ub}
There is $\gamma \geq 1$ such that $\exwdisc_p(n) \leq \gamma \cdot \sqrt{n}$ for all $n \in \N, p \in [0, 1]$.
\end{theorem}

The following is a simple but useful observation:
\begin{observation}[\cite{DoerrSr03}] \label{obs:color-v-weighted}
For any $k, n \in \N$, $\exdisc(n, k) \geq \exwdisc_{1/k}(n)$
\end{observation}

\section{Proof of the Main Theorem}

We recall the following well known construction in discrepancy literature: Let $\bW$ be $\frac{1}{2}(\bone_{n \times n} + \bH)$ where $\bone_{n \times n}$ is an $n \times n$ all-$1$ matrix and $\bH$ is an $n \times n$ Hadamard matrix. This construction appeared as early as \cite{LovaszSpVe86} (based on~\cite{Spencer85}) and it yields the $\Omega(\sqrt{n})$ lower bound for $k = 2$.

For our purpose, it will be convenient to employ the following property (where $\bW$ is as defined above), which was shown in \cite{Chazelle01} and stated more explicitly in \cite{CharikarNeNi11}.

\begin{lemma}[\cite{Chazelle01,CharikarNeNi11}] \label{lem:hadamard-discrepancy}
For any $n \in \N$ that is a power of two, there exists $\bW \in \{0, 1\}^{n \times n}$ such that $\|\bW\bz\|_2^2 \geq \frac{n}{4} \cdot \left(\sum_{i=2}^n \bz_i^2\right)$ for all $\bz \in \R^n$.
\end{lemma}

We are now ready to prove our main result, an $\Omega(\sqrt{n})$ lower bound on $\exwdisc_p(n)$ for any $p \in (0, 1)$. This improves on $\Omega(\sqrt{p n})$ lower bound from \cite{ManurangsiS22}, and is asymptotically tight in light of \autoref{thm:wdisc-ub}.

\begin{theorem} \label{thm:main-weighted}
For any $p \in (0, 1), n \in \N$, we have $\exwdisc_p(n) \geq \sqrt{(n - 1)}/16$.
\end{theorem}

As alluded to in the introduction, our construction is very simple: Just take $\bA$ to be $\Theta(1/p)$ copies of $\bW$ horizontally stacked together! (Equivalently, the corresponding hypergraph of $\bA$ has $\Theta(1/p)$ copies of each vertex in $\bW$ and each new hyperedge contains all copies of the vertices belonging to the original hyperedge.) This is formalized below.

\begin{proof}[Proof of Theorem~\ref{thm:main-weighted}]
Since $\exwdisc_p(n)$ is non-decreasing in $n$, it suffices to prove $\exwdisc_p(n) \geq \sqrt{(n - 1)}/8 =: \Delta$ for all $n \in \N$ that is a power of two. Furthermore, it is obvious that $\exwdisc_p(n) = \exwdisc_{1 - p}(n)$. Thus, we may assume that $p 
\leq 1/2$.

Let $t = \lfloor 1/2p \rfloor$ and $\bW \in \{0, 1\}^{n \times n}$ be the matrix guaranteed by \autoref{lem:hadamard-discrepancy}. We construct the matrix $\bA \in \{0,1\}^{n \times nt}$ as $\bA := [\bW \mid \cdots \mid \bW]$ where there are $t$ blocks of $\bW$.

Consider any $\bx \in \{0, 1\}^{nt}$. Let $\bz \in \{0, \dots, t\}^{n}$ be defined by $\bz_i = \sum_{j=0}^{t - 1} \bx_{i+nj}$. It is simple to see that
$\bA(p \cdot \bone - \bx) = \bW(pt \cdot \bone - \bz).$ 
Notice, by our choice of $t$, that $1/4 \leq pt \leq 1/2$. As a result, $pt \cdot \bone - \bz$ is a vector such that every coordinate has an absolute value at least $1/4$. Applying \autoref{lem:hadamard-discrepancy}, we thus have
$\|\bA(p \cdot \bone - \bx)\|^2_2 = \|\bW(pt \cdot \bone - \bz)\|^2_2 \geq \frac{n}{4} \cdot (n-1)/16 = n\Delta^2.$
This implies that $\|\bA(p \cdot \bone - \bx)\|_\infty \geq \Delta$. Thus, $\wdisc_p(\bA) \geq \Delta$, which concludes our proof.
\end{proof}

\autoref{obs:color-v-weighted} and \autoref{thm:main-weighted} together yield our main lower bound for $\exdisc(n, k)$:

\begin{corollary} \label{cor:main-multicolor}
For any $n, k \in \N$ such that $k \geq 2$, we have $\exdisc(n, k) \geq \sqrt{(n - 1)}/16$.
\end{corollary}

\section{Applications to Group Fair Division}
\label{sec:group-fair}

We consider the following setting of (approximate) fair division of indivisible goods for groups of agents, which has been studied in a number of recent works~\cite{ManurangsiSu17,Suksompong18,Suksompong18-2,SegalhaleviSu19,KyropoulouSuVo20,ManurangsiS22,ManurangsiS24,caragiannis2025newlowerboundmulticolor}.
There is a set $G$ of $m$ items and $k$ groups of agents of sizes $n_1, \dots, n_k$. We use $n$ to denote $n_1 + \cdots + n_k$. Let $a^{(i, j)}$ denote the $j$-th agent in group $i$ and $u^{(i, j)}$ be her utility function; we assume that these utilities are monotone and additive\footnote{This means that $u^{(i, j)}(S) = \sum_{g \in S} u^{(i, j)}(g)$ for all $S \subseteq G$, and that $u^{(i, j)}(g) \geq 0$ for all $g \in G$.}. An allocation $A = (A_1, \dots, A_k)$ is a partition of $G$ into $k$ bundles where the bundle $A_i$ is allocated to group $i$. We study three notions of approximate fairness, similar to~\cite{ManurangsiS22}:
\begin{itemize}
\item An allocation $A$ is \emph{envy-free up to $c$ goods (EF$c$)} if, for all $i, i' \in [k], j \in [n_i]$, there exists a set $B \subseteq G$ of size at most $c$ such that $u^{(i,j)}(A_i) \geq u^{(i,j)}(A_{i'} \setminus B)$.
\item An allocation $A$ is \emph{proportional up to $c$ goods (PROP$c$)} if, for all $i \in [k], j \in [n_i]$, there exists a set $B \subseteq G \setminus A_i$ of size at most $c$ such that $u^{(i,j)}(A_i) \geq u^{(i,j)}(G) / k - u^{(i,j)}(B)$.
\item  An allocation $A$ is \emph{a consensus $1/k$-division up to $c$ goods\footnote{This notion does not depend on how the users are divided across the different groups.} (CD$_kc$)} if, for all $i, i', i'' \in [k], j \in [n_i]$, there exists a set $B \subseteq G$ of size at most $c$ such that $u^{(i,j)}(A_{i'}) \geq u^{(i,j)}(A_{i''} \setminus B)$.
\end{itemize}

We define $\cef(n_1, \dots, n_k)$ (resp., $\cprop(n_1, \dots, n_k)$) to be the smallest value of $c$ such that an EF$c$ (resp., PROP$c$) allocation always exists.
Let $\ccdiv(n)$ denote the analogous value for CD$_kc$ (which only depend on the total number of agents $n$ but not on $n_1, \dots, n_k$). Since these quantities are independent of the orders of $n_1, \dots, n_k$, we will assume throughout that $n_1 \geq \cdots \geq n_k \geq 1$ which will allow us to state the results in a more convenient manner.

\paragraph*{Our Results.}
We provide improved lower bounds for all three quantities, as stated below.

\begin{corollary} \label{cor:ef}
For positive integers $k \geq 2$ and $n_1 \geq \dots \geq n_k$, $\cef(n_1, \dots, n_k) \geq \Omega(\sqrt{n_1})$.
\end{corollary}

\begin{corollary} \label{cor:prop}
For positive integers $k \geq 2$ and $n_1 \geq \dots \geq n_k$,
\[\cprop(n_1, \dots, n_k) \geq \Omega\left(\max_{i^* \in [k]} \frac{i^*}{k} \cdot \sqrt{n_{i^*}}\right).\]
\end{corollary}

\begin{corollary} \label{cor:cd}
For positive integers $k \geq 2$ and $n \geq k$, $\ccdiv(n) \geq \Omega(\sqrt{n})$.
\end{corollary}

\begin{table}
\begin{center}
{\setlength\extrarowheight{15pt}
\begin{tabular}{ |c|c|c|c|} 
 \hline
  Quantity & Our Lower Bounds & Bounds in \cite{caragiannis2025newlowerboundmulticolor} & Bounds in \cite{ManurangsiS22}  \\ 
  \hline
 $\cef(n_1, \dots, n_k)$ & $\Omega(\sqrt{n_1})$ & $\Omega\left(\sqrt{\frac{n}{k \log k}}\right)$ & $\Omega\left(\sqrt{\frac{n_1}{k^3}}\right)$ \\ 
 $\cprop(n_1, \dots, n_k)$ & $\Omega\left(\max_{i^* \in [k]} 
 \frac{i^*}{k} \cdot \sqrt{n_{i^*}}\right)$ & $\Omega\left(\sqrt{\frac{n}{k^3 \log k}}\right)$, $\Omega\left(\sqrt{\frac{n_k}{\log k}}\right)$ & $\Omega\left(\sqrt{\frac{n_1}{k^3}}\right)$ \\ 
  $\ccdiv(n)$ & $\Omega(\sqrt{n})$ & $\Omega\left(\sqrt{\frac{n}{\log k}}\right)$ & $\Omega\left(\sqrt{\frac{n}{k}}\right)$ \\
 \hline
\end{tabular}
}
\end{center}
\caption{Summary of our lower bounds for group fair division and known lower bounds. Here, we assume $n_1 \geq \cdots \geq n_k \geq 1$. As discussed below, our lower bound is a strict improvement over previous lower bounds by a factor of $\Omega(\sqrt{\log k})$ for \emph{all} values of $n_1, \dots, n_k$. Moreover, the improvement for $\cef$ and $\cprop$ can be as large as $\Omega(\sqrt{k \log k})$ and $\Omega(k\sqrt{\log k})$ respectively for \emph{some} values of $n_1, \dots, n_k$. Finally, note also that the previous three lower bounds for $\cprop$ are incomparable.} \label{tab:group-fair-bounds}
\end{table}

Our lower bounds and those from \cite{ManurangsiS22,caragiannis2025newlowerboundmulticolor} are summarized in \autoref{tab:group-fair-bounds}. 
Our lower bound for $\ccdiv(n)$ is asymptotically tight and improves upon the $\Omega(\sqrt{n/\log k})$ lower bound from \cite{caragiannis2025newlowerboundmulticolor}. Our lower bound for $\cef(n_1, \dots, n_k)$ improves upon the $\Omega\left(\sqrt{\frac{n}{k \log k}}\right)$ lower bound from \cite{caragiannis2025newlowerboundmulticolor} since $n_1 \geq n/k$; moreover, in the case where $n_1 = \Omega(n)$, our improvement can be as large as $\Omega(\sqrt{k \log k})$. Finally, our lower bound $\cprop(n_1, \dots, n_k)$ is also an improvement over all three (incomparable) lower bounds from~\cite{ManurangsiS22,caragiannis2025newlowerboundmulticolor} by a factor of at least $\Omega(\sqrt{\log k})$ in all regimes of parameters; moreover, when e.g.\ $n_1, \cdots, n_{\lfloor k/2 \rfloor} = \Theta(n/k)$ and $n_{\lfloor k/2 \rfloor + 1}, \dots, n_k = O(1)$, our lower bound of $\Omega\left(\sqrt{\frac{n}{k}}\right)$ is an $\Omega(k\sqrt{\log k})$ factor improvement over the previous bounds.

\subsection{Consensus \texorpdfstring{$1/k$}{1/k}-Division}
In~\cite{ManurangsiS22}, it was shown that $\cef(n_1, \dots, n_k), \cprop(n_1, \dots, n_k)$ and $\ccdiv(n)$ are all at most $O(\sqrt{n})$. Furthermore, they provide a lower bound through the following:
\begin{theorem}[\cite{ManurangsiS22}] \label{thm:known-red}
For positive integers $k \geq 2$ and $n_1 \geq \dots \geq n_k$ (with $n = n_1 + \cdots + n_k$),
\begin{itemize}
\item $\cef(n_1, \dots, n_k) \geq \exwdisc_{1/k}\left(\lfloor n_1/2\rfloor\right) / k$
\item $\cprop(n_1, \dots, n_k) \geq \exwdisc_{1/k}\left(\lfloor n_1/2\rfloor\right) / k$
\item $\ccdiv(n) \geq \exdisc(n, k)$.
\end{itemize}
\end{theorem}

Combining \autoref{cor:main-multicolor} with the above theorem, we immediately arrive at the tight lower bound for $\ccdiv(n)$ (\autoref{cor:cd}).

\subsection{Improved Bounds for EF and PROP}

Interestingly, however, applying the known reduction (\autoref{thm:known-red}) to our bound on $1/k$-weighted discrepancy (\autoref{thm:main-weighted}) does \emph{not} yet yield strong lower bounds for $\cef$ and $\cprop$. In fact, the bound we would have gotten for $\cef$ is strictly worse than that from \cite{caragiannis2025newlowerboundmulticolor} in all regimes of parameters. To overcome this barrier, we give better lower bounds for $\cef$ and $\cprop$ based on $1/k$-weighted discrepancy, as stated below.

\begin{theorem} \label{thm:ef-reduction}
For positive integers $k \geq 2$ and $n_1 \geq \dots \geq n_k$, $\cef(n_1, \dots, n_k) \geq \frac{\exwdisc_{1/k}\left(\lfloor \frac{n_1}{2}\rfloor\right)}{2}$.
\end{theorem}

\begin{theorem} \label{thm:prop-reduction}
For positive integers $k \geq 2$ and $n_1 \geq \dots \geq n_k$, \begin{align*}
\cprop(n_1, \dots, n_k) \geq \max_{i^* \in [k]} \frac{i^*}{k} \cdot \exwdisc_{1/k}\left(\left\lfloor \frac{n_{i^*}}{2}\right\rfloor\right).
\end{align*}
\end{theorem}

Combining these theorems with \autoref{thm:main-weighted} immediately yields Corollaries~\ref{cor:ef} and \ref{cor:prop}.

We note that our lower bound in \autoref{thm:prop-reduction} for $\cprop$ is a generalization of that from \cite{ManurangsiS22}, which corresponds to the case $i^* = 1$. Meanwhile, our lower bound for $\cef$ in \autoref{thm:ef-reduction} is a direct improvement of that from \cite{ManurangsiS22} by a factor of $\Omega(k)$.

In our constructions, the utility of each item will always be at most 1. Throughout this section, we will assume this without stating it explicitly. We say that a utility $u$ is a \emph{complement} of a utility $u'$ iff $u(g) = 1 - u'(g)$ for all $g \in G$. Similarly, an agent is a complement of another if their utilities are complements. We start with the following lemma, which will be useful in subsequent proofs. This lemma is a generalization of a property shown in \cite[Theorem 4.1]{ManurangsiS22}.

\begin{lemma} \label{lem:prop-to-disc}
Let $A = (A_1, \dots, A_k)$ be any PROP$c$ allocation. Suppose that a group $i$ contains two agents $j, j'$ that are complements. Then, we have
\begin{align*}
\left|u^{(i, j)}(A_i) - \frac{u^{(i,j)}(G)}{k}\right| \leq c + \left[|A_i| - \frac{m}{k}\right]_+.
\end{align*}
\end{lemma}

\begin{proof}
Since $A$ is PROP$c$ and the valuation of each item is at most one, we have
$u^{(i, j)}(A_i) \geq \frac{u^{(i,j)}(G)}{k} - c$  and $u^{(i, j')}(A_i) \geq \frac{u^{(i,j')}(G)}{k} - c.$
The former is equivalent to
$u^{(i, j)}(A_i) - \frac{u^{(i,j)}(G)}{k} \geq -c.$
Moreover, Since $u^{(i, j')}$ is a complement to $u^{(i, j)}$, the latter can be rearranged as
\begin{align*}
c &\geq \frac{m - u^{(i,j)}(G)}{k} - \left(|A_i| - u^{(i, j)}(A_i)\right)
= \left(\frac{m}{k} - |A_i|\right) + u^{(i, j)}(A_i) - \frac{u^{(i,j)}(G)}{k}.
\end{align*}
Combining these two inequalities yields the claim.
\end{proof}

\paragraph*{Construction for PROP$c$.}
\autoref{lem:prop-to-disc} provides a clear intuition for the construction: We use the agents' utilities in the first $i^*$ groups to encode the matrix (and its complement). In any PROP$c$ allocation $A = (A_1, \dots, A_k)$, if one of the first $i^*$ groups has small bundle size (only slightly above $m/k$), then the lemma immediately allows us to relate $c$ back to $1/k$-weighted discrepancy. To ensure the condition, we have an agent with all-1 utility in each of the remaining $k - i^*$ groups; this ensures that the bundle these groups receive cannot be too small, which implies that one of the first $i^*$ groups has a small bundle.

We note that our construction and our bound indeed coincide with that of \cite{ManurangsiS22} in the case of $i^* = 1$. Thus, ours may be viewed as a generalization of their construction.

\begin{proof}[Proof of Theorem \ref{thm:prop-reduction}]
Suppose for the sake of contradiction that, for some $i^* \in [k]$, $c := \cprop(n_1, \dots, n_k) < \frac{i^*}{k} \cdot \Delta$ where $\Delta = \exwdisc_{1/k}(n')$ and $n' = \lfloor n_{i^*}/2 \rfloor$.

Let $\bA \in [0, 1]^{n' \times m}$ be such that $\wdisc_{1/k}(\bA) \geq \Delta$. Let $G = [m]$, and the utilities of the agents be as follows. First, for groups $i \in \{i^* + 1, \dots, k\}$, we let their first agent have an all-1 utility, i.e.\ $u^{(i,1)}(g) = 1$ for all $g \in [m]$. For each of the remaining group $i \in [i^*]$ and for all $j \in [n']$, we set the utility for the $j$-th and $(j + n')$-th agents as follows:\footnote{The remaining agents' utilities can be set arbitrarily.}
\begin{align*}
u^{(i,j)}(g) = \bA_{j,g} & & u^{(i,j+n')}(g) = 1 - \bA_{j,g} & &\forall g \in [m].
\end{align*}

Let $A = (A_1, \dots, A_k)$ be an allocation that is PROP$c$. For all $i \in [i^*]$, since $u^{(i,j)}(g)$ and $u^{(i,j+n')}(g)$ are complement, we can apply \autoref{lem:prop-to-disc} to conclude that
\begin{align*}
c + \left[|A_i| - \frac{m}{k}\right]_+ \geq \left|u^{(i, j)}(A_i) - \frac{u^{(i,j)}(G)}{k}\right| = \left|\left(\bA\left(\frac{1}{k}\cdot\bone - \bone(A_i)\right)\right)_j\right| & &\forall j \in [n'].
\end{align*} 
Since $\wdisc_{1/k}(\bA) \geq \Delta$, we have $|A_i| - \frac{m}{k} \geq \Delta - c > \left(\frac{k}{i^*} - 1\right) \cdot c$. Since $|A_1| + \cdots + |A_k| = m$ and $|A_1|, \dots, |A_{i^*}| > \frac{m}{k} + \left(\frac{k}{i^*} - 1\right) \cdot c$, there exists $i \in \{i^* + 1, \dots, k\}$ with $|A_i| < \frac{m}{k} - c$. However, this means that $u^{(i,1)}(A_i) < \frac{u^{(i,1)}(G)}{k} - c$. Thus, $A$ cannot be PROP$c$, a contradiction.
\end{proof}

\paragraph*{Construction for EF$c$.} Our construction for EF$c$ is in fact the same as that of PROP in the case $i^* = 1$ above (which is also the same as that of \cite{ManurangsiS22}). The difference is that, we use EF$c$ to ensure that $|A_1|$ is smaller than $m/k + c$.
This is formalized below.

\begin{proof}[Proof of Theorem \ref{thm:ef-reduction}]
Let $n' := \lfloor n_1/2 \rfloor$ and $\Delta = \exwdisc_{1/k}(n')$. Suppose for the sake of contradiction that $c := \cef(n_1, \dots, n_k) < \Delta/2$.

Let $\bA \in [0, 1]^{n' \times m}$ be such that $\wdisc_{1/k}(\bA) \geq \Delta$. Let $G = [m]$, and the utilities of the agents be as follows. First, for groups $i \in \{2, \dots, k\}$, we let their first agent have an all-1 utility, i.e.\ $u^{(i,1)}(g) = 1$ for all $g \in [m]$. For the first group and for all $j \in [n']$, we set the utility for the $j$-th and $(j + n')$-th agents as follows:
\begin{align*}
u^{(1,j)}(g) = \bA_{j,g} & & u^{(1,j+n')}(g) = 1 - \bA_{j,g} & &\forall g \in [m].
\end{align*}

Let $A = (A_1, \dots, A_k)$ be an EF$c$ allocation. Since any EF$c$ allocation is also PROP$c$, we can apply \autoref{lem:prop-to-disc} to conclude that
\begin{align*}
c + \left[|A_1| - \frac{m}{k}\right]_+ \geq \left|u^{(1, j)}(A_1) - \frac{u^{(1,j)}(G)}{k}\right| = \left|\left(\bA\left(\frac{1}{k}\cdot\bone - \bone(A_1)\right)\right)_j\right| & &\forall j \in [n'].
\end{align*}
Since $\wdisc_{1/k}(\bA) \geq \Delta$, we have $|A_1| - \frac{m}{k} \geq \Delta - c > c$. Since $|A_1| + \cdots + |A_k| = m$ and $|A_1| > m/k + c$, there exists $i \in \{2, \dots, k\}$ with $|A_i| < m/k$. However, this means that $u^{(i,1)}(A_i) < u^{(i,1)}(A_1) - c$. Thus, $A$ cannot be EF$c$, a contradiction.
\end{proof}

\section{Improved Upper Bound on PROP}
\label{sec:prop-ub}

We will prove the following bound on $\cprop$, which improves on the $O(\sqrt{n})$ bound of \cite{ManurangsiS22}.
\begin{corollary} \label{cor:prop-ub-main}
For positive integers $k \geq 2$ and $n_1 \geq \dots \geq n_k$, $\cprop(n_1, \dots, n_k) \leq O(\sqrt{n_1})$.
\end{corollary}
An interesting case to highlight is when $n_1 = \cdots = n_k = n/k$. In this case, \autoref{cor:prop-ub-main} yields an upper bound of $O(\sqrt{n/k})$ which matches our lower bound from \autoref{cor:prop}.

\paragraph*{Asymmetric Discrepancy.} Similar to previous proofs, we will prove \autoref{cor:prop-ub-main} by relating it to a discrepancy-based notion. However, to achieve the improved bound, we need a new notion of discrepancy, which we define as \emph{asymmetric discrepancy} below in \autoref{def:k-odisc}. At a high level, this is similar to discrepancy except that there are $k$ different matrices that is used to measure the discrepancy of different colors.

\begin{definition} \label{def:k-odisc}
The \emph{asymmetric $k$-color discrepancy} of $k$ matrices $\bA^1 \in [0, 1]^{n_1 \times m}, \dots, \bA^k \in [0, 1]^{n_k \times m}$ is defined as
\begin{align*}
\odisc(\bA^1, \dots, \bA^k) := \min_{\chi: [m] \to [k]} \max_{s \in [k]} \left\|\bA^s\left(\frac{1}{k} \cdot \bone - \bone(\chi^{-1}(s))\right)\right\|_{\infty}.
\end{align*}
For all $k, n_1, \dots, n_k \in \N$, let 
$$\exodisc(n_1, \dots, n_k) := \sup_{m \in \N} \sup_{\bA^1 \in [0, 1]^{n_1 \times m}, \dots, \bA^k \in [0, 1]^{n_k \times m}} \odisc(\bA^1, \dots, \bA^k).$$
\end{definition}

Observant readers can probably see similarities between the above definition and that of proportionality already. Indeed, we can formally show that an upper bound on the former implies that of the latter, as stated below.

\begin{theorem} \label{thm:odisc-to-prop}
For all $k, n_1, \dots, n_k \in \N$, $\cprop(n_1, \dots, n_k) \leq 2 \cdot \lceil \exodisc(n_1, \dots, n_k) \rceil$.
\end{theorem}

Finally, via known techniques~\cite{DoerrSr03}, we  give the following upper bound on the asymmetric discrepancy:

\begin{theorem} \label{thm:odisc-ub}
For positive integers $k$ and $n_1 \geq \dots \geq n_k$, $\exodisc(n_1, \dots, n_k) \leq O(\sqrt{n_1})$
\end{theorem}

Notice that our main upper bound (\autoref{cor:prop-ub-main}) is indeed an immediate corollary of the above two theorems.
The subsequent two subsections are devoted to the proofs of \autoref{thm:odisc-ub} and \autoref{thm:odisc-to-prop}, respectively.

\subsection{Upper Bound on Asymmetric Discrepancy}

To prove \autoref{thm:odisc-ub}, we employ the \enquote{recursion} construction similar to \cite{DoerrSr03}. Namely, we group the colors $[k]$ into two groups $\{1, \dots, k_1\}$ and $\{k_1 + 1, \dots, k\}$ where $k_1 = \lfloor k/2 \rfloor$. We then use the weighted discrepancy upper bound (\autoref{thm:wdisc-ub}) to partition the columns into two groups such that the first group is roughly $\frac{k_1}{k}$ fraction w.r.t.\ each row. We then recurse on the two sides. The improvement here is due to the fact that, when we recurse on the first group, we only have to consider the first $k_1$ matrices $\bA^1, \dots, \bA^{k_1}$ instead of all the $k$ matrices. This intuition is formalized below.

\begin{proof}[Proof of Theorem \ref{thm:odisc-ub}]
Let $\zeta := 100\gamma$ where $\gamma$ is the constant from \autoref{thm:wdisc-ub}.
We will prove the following statement by induction on $k'$:
\begin{align} \label{eq:prop-induction}
\exodisc(n_1, \dots, n_{k'}) \leq \zeta \left(1 - \frac{1}{\sqrt{k'}}\right) \cdot \sqrt{n_1} & &\forall n_1 \geq \cdots \geq n_{k'}.
\end{align}
The base case $k' = 1$ is obvious. For the inductive step, consider any integer $k \geq 2$. Suppose that \eqref{eq:prop-induction} holds for any $k' < k$. We will show that it also holds for $k' = k$. For any matrices $\bA^1 \in [0, 1]^{n_1 \times m}, \dots, \bA^k \in [0, 1]^{n_k \times m}$, we create our $k$-coloring $\chi: [m] \to [k]$ as follows:
\begin{enumerate}
\item Let $k_1 = \lfloor k/2 \rfloor$ and $k_2 = k - k_1 = \lceil k/2 \rceil$.
\item Let $\tbA \in [0, 1]^{(n_1 + \cdots + n_k) \times m}$ be the matrix resulting from vertically concatenating $\bA^1, \dots, \bA^k$. Then, apply \autoref{thm:wdisc-ub} to $\tbA$ with $p = \frac{k_1}{k}$ to obtain $\bx \in \{0, 1\}^m$ such that
\begin{align} \label{eq:outer-coloring}
\|\tbA(p \cdot \bone - \bx)\|_{\infty} \leq \gamma \sqrt{n} \leq \gamma \sqrt{n_1 k}.
\end{align}
Let $S_1 := \{i \in [m] \mid \bx_i = 1\}$ and $S_2 := [m] \setminus S_1$. Notice that $\bx = \bone(S_1)$ and $\bone - \bx = \bone(S_2)$.
\item Apply the inductive hypothesis to $\bA^1, \dots, \bA^{k_1}$ to obtain $\tchi: S_1 \to [k_1]$ such that
\begin{align} \label{eq:first-group-inner-coloring}
\left\|\bA^s\left(\frac{1}{k_1} \cdot \bone(S_1) - \bone(\tchi^{-1}(s))\right)\right\|_{\infty} \leq \zeta\left(1 - \frac{1}{\sqrt{k_1}}\right) \cdot \sqrt{n_1}. & &\forall s \in [k_1]
\end{align}
Similarly, apply the inductive hypothesis to $\bA^{k_1 + 1}, \dots, \bA^{k}$ to obtain a coloring $\ochi: S_2 \to \{k_1+1, \dots, k\}$ such that
\begin{align} \label{eq:second-group-inner-coloring}
\left\|\bA^s\left(\frac{1}{k_2} \cdot \bone(S_2) - \bone(\ochi^{-1}(s))\right)\right\|_{\infty} \leq \zeta\left(1 - \frac{1}{\sqrt{k_2}}\right) \cdot \sqrt{n_1}. & &\forall s \in \{k_1+1,\dots,k\}
\end{align}
Set the final coloring $\chi$ to be the concatenation of the two colorings $\tchi, \ochi$.
\end{enumerate}

We next bound $\left\|\bA^s\left(\frac{1}{k} \cdot \bone - \bone(\chi^{-1}(s))\right)\right\|_{\infty}$ for each $s \in [k]$ based on whether $s \in [k_1]$.

First, consider the case where $s \in [k_1]$. In this case, we have
\begin{align*}
&\left\|\bA^s\left(\frac{1}{k} \cdot \bone - \bone(\chi^{-1}(s))\right)\right\|_{\infty} \\
&\leq \left\|\bA^s\left(\frac{1}{k_1} \cdot \bone(S_1) - \bone(\tchi^{-1}(s))\right)\right\|_{\infty} + \left\|\bA^s\left(\frac{1}{k} \cdot \bone - \frac{1}{k_1} \cdot \bone(S_1)\right)\right\|_{\infty} \\
&= \left\|\bA^s\left(\frac{1}{k_1} \cdot \bone(S_1) - \bone(\tchi^{-1}(s))\right)\right\|_{\infty} + \frac{1}{k_1} \cdot \left\|\bA^s\left(p \cdot \bone - \bx\right)\right\|_{\infty} \\
&\overset{\eqref{eq:outer-coloring}, \eqref{eq:first-group-inner-coloring}}{\leq} \zeta\left(1 - \frac{1}{\sqrt{k_1}}\right) \cdot \sqrt{n_1} + \frac{1}{k_1} \cdot \gamma \sqrt{n_1 k} \\
&= \zeta\left(1 - \frac{1}{\sqrt{k_1}} + \frac{0.01\sqrt{k}}{k_1}\right) \cdot \sqrt{n_1}
\leq \zeta\left(1 - \frac{1}{\sqrt{k_1}} + \frac{0.01\sqrt{3k_1}}{k_1}\right) \cdot \sqrt{n_1} \\*
&\leq \zeta\left(1 - \frac{0.9}{\sqrt{k_1}}\right) \cdot \sqrt{n_1}
\leq \zeta\left(1 - \frac{1}{\sqrt{k}}\right) \cdot \sqrt{n_1}.
\end{align*}
Similarly, for the case $s \notin [k_1]$, we have
\begin{align*}
&\left\|\bA^s\left(\frac{1}{k} \cdot \bone - \bone(\chi^{-1}(s))\right)\right\|_{\infty} \\*
&\leq \left\|\bA^s\left(\frac{1}{k_2} \cdot \bone(S_2) - \bone(\ochi^{-1}(s))\right)\right\|_{\infty} + \left\|\bA^s\left(\frac{1}{k} \cdot \bone - \frac{1}{k_2} \cdot \bone(S_2)\right)\right\|_{\infty} \\
&= \left\|\bA^s\left(\frac{1}{k_2} \cdot \bone(S_2) - \bone(\ochi^{-1}(s))\right)\right\|_{\infty} + \frac{1}{k_2} \cdot \left\|\bA^s\left(p \cdot \bone - \bx\right)\right\|_{\infty} \\
&\overset{\eqref{eq:outer-coloring}, \eqref{eq:second-group-inner-coloring}}{\leq} \zeta\left(1 - \frac{1}{\sqrt{k_2}}\right) \cdot \sqrt{n_1} + \frac{1}{k_2} \cdot \gamma \sqrt{n_1 k} \\
&= \zeta\left(1 - \frac{1}{\sqrt{k_2}} + \frac{0.01\sqrt{k}}{k_2}\right) \cdot \sqrt{n_1}
\leq \zeta\left(1 - \frac{1}{\sqrt{k_2}} + \frac{0.01\sqrt{2k_2}}{k_2}\right) \cdot \sqrt{n_1} \\*
&\leq \zeta\left(1 - \frac{0.9}{\sqrt{k_2}}\right) \cdot \sqrt{n_1}
\leq \zeta\left(1 - \frac{1}{\sqrt{k}}\right) \cdot \sqrt{n_1}.
\end{align*}
Thus, in both cases, we have $\left\|\bA^s\left(\frac{1}{k} \cdot \bone - \bone(\chi^{-1}(s))\right)\right\|_{\infty} \leq \zeta\left(1 - \frac{1}{\sqrt{k}}\right) \cdot \sqrt{n_1}$, and we can conclude that \eqref{eq:prop-induction} holds for all $k' \in \N$ as desired.
\end{proof}

\subsection{From Asymmetric Discrepancy to Proportionality}

Finally, we prove \autoref{thm:odisc-to-prop} via a similar reduction from \cite{ManurangsiS22} who used it to relate discrepancy with (approximate) consensus division. The rough idea of their reduction is that, for each agent, we divide the items into \enquote{large} and \enquote{small} based on the utility values. Then, we create two discrepancy constraints: The first is to ensure that the \emph{number} of large items are balanced between each part, and the second is to ensure that the \emph{utility} of small items are balanced. The latter involves scaling the utility values with an appropriate normalization so that the values belong to $[0,1]$. We use almost the same reduction as this, except that the reduction is even simpler. Specifically, we do not need the first type of constraints at all. This is because we seek (approximate) proportionality, which is weaker than (approximate) consensus division.

\begin{proof}[Proof of Theorem \ref{thm:odisc-to-prop}]
The case $k = 1$ is obvious; henceforth, we will focus on $k \geq 2$.
Let $H := \lceil \exodisc(n_1, \dots, n_k) \rceil$. Assume w.l.o.g.\ that $G = [m]$ and\footnote{Otherwise, we may add dummy goods with zero values to increase $m$.} $m \geq kH$.

For every agent $a^{(i, j)}$, let $\Slarge^{(i,j)} \subseteq G$ be the set of her $kH$ most valuable goods, and $\tSlarge^{(i,j)} \subseteq \Slarge^{(i,j)}$ be the set of her $2H$ most valuable goods. Furthermore, let $p^{(i, j)} := \min_{g \in \Slarge^{(i, j)}} u^{(i, j)}(g)$, and $\bz^{(i,j)} \in [0, 1]^m$ be defined as\footnote{We use the convention that 0/0 = 0.}
\begin{align*}
\bz^{(i, j)}_g =
\begin{cases}
u^{(i, j)}(g) / p^{(i, j)} & \text{ if } g \notin \Slarge^{(i, j)}, \\
0 & \text{ otherwise.}
\end{cases}
& &\forall g \in [m].
\end{align*}

For each group $i \in [k]$, let $\bA^i = \begin{bmatrix}
\bz^{(i, 1)}
\cdots
\bz^{(i, n_i)}
\end{bmatrix}^T$. We have $\bA^i \in [0, 1]^{n_i \times m}$. From the definition of $\exodisc$, there exists $\chi: [m] \to [k]$ such that
\begin{align} \label{eq:prop-disc-bound}
\left\|\bA^i\left(\frac{1}{k} \cdot \bone - \bone(\chi^{-1}(i))\right)\right\|_{\infty} \leq H & &\forall i \in [k].
\end{align}
Consider the allocation $A = (A_1, \dots, A_k)$ where $A_i = \chi^{-1}(i)$ for all $i \in [k]$.
We claim that this allocation is PROP$c$ for $c = 2H$. To show this, consider any agent $a^{(i, j)}$. From \eqref{eq:prop-disc-bound}, we have
\begin{align}
H &\geq \left<\bz^{(i, j)}, \frac{1}{k} \cdot \bone - \bone(A_i)\right> = \frac{1}{p^{(i, j)}} \left(\frac{u^{(i, j)}\left(G \setminus \Slarge^{(i, j)}\right)}{k} - u^{(i, j)}\left(A_i \setminus \Slarge^{(i, j)}\right)\right). \label{eq:prop-util-bound}
\end{align}
Let $B = \tSlarge^{(i, j)} \setminus A_i$. Since $|\tSlarge^{(i, j)}| = c$, we have $|B| \leq c$. Furthermore, we have
\begin{align*}
&u^{(i,j)}(A_i) + u^{(i,j)}(B) - u^{(i,j)}(G) / k \\
&\geq u^{(i, j)}\left(\tSlarge^{(i, j)}\right) + u^{(i, j)}\left(A_i \setminus \Slarge^{(i, j)}\right) - u^{(i,j)}(G) / k \\
&= \left(u^{(i, j)}\left(\tSlarge^{(i, j)}\right) - \frac{u^{(i, j)}\left(\Slarge^{(i,j)}\right)}{k}\right) - \left(\frac{u^{(i, j)}\left(G \setminus \Slarge^{(i, j)}\right)}{k} - u^{(i, j)}\left(A_i \setminus \Slarge^{(i, j)}\right)\right) \\
&\overset{\eqref{eq:prop-util-bound}}{\geq} \left(1 - \frac{1}{k} \cdot \frac{|\Slarge^{(i, j)}|}{|\tSlarge^{(i, j)}|}\right) \cdot u^{(i, j)}\left(\tSlarge^{(i, j)}\right) - H \cdot p^{(i, j)} \\
&= \frac{1}{2} \cdot u^{(i, j)}\left(\tSlarge^{(i, j)}\right) - H \cdot p^{(i, j)} \\
&\geq \frac{1}{2} \cdot |\tSlarge^{(i, j)}| \cdot p^{(i, j)} - H \cdot p^{(i, j)} \\
&= 0,
\end{align*}
where the last inequality is from our definition of $p^{(i, j)}$.
Thus, the allocation is PROP$c$.
\end{proof}

\section{Conclusion and Open Questions}

We prove an asymptotically tight lower bound of $\Omega(\sqrt{n})$ on both $\exdisc(n, k)$ and $\exwdisc_p(n)$, which in turn implies an asymptotically tight  bound for $\ccdiv(n)$. We also obtain improved lower bounds for $\cef(n_1, \dots, n_k)$ and $\cprop(n_1, \dots, n_k)$, although they are not yet tight. It remains an interesting question to close this gap. A particularly natural case is when $n_1 = \cdots = n_k = n/k$; our lower bounds for $\cef$ is $\Omega(\sqrt{n/k})$ whereas the upper bound from~\cite{ManurangsiS22} is $O(\sqrt{n})$. Intriguingly, for $n_1 = \cdots = n_k = 1$, it is known that $\cef(n_1, \dots, n_k) = 1$~\cite{LiptonMaMo04}. Thus, closing this gap seems to require some innovation on the \emph{upper bound} front.

\subsection*{Acknowledgement}
We thank the anonymous reviewers for SOSA 2026 and TheoretiCS who have provided some feedback, including suggesting a generalized notion used in \autoref{def:k-odisc}.

\printbibliography
\end{document}